\documentclass[11pt]{article}
\pdfoutput=1
\usepackage{amsmath,amssymb,amsthm,color,graphicx}
\usepackage[utf8]{inputenc}
\usepackage[T1]{fontenc}
\usepackage{microtype}
\usepackage[disable]{todonotes} 
\usepackage{thmtools}
\usepackage{thm-restate}
\usepackage[subrefformat=parens,labelformat=parens]{subcaption}
\usepackage{booktabs}
\usepackage[top=1in, bottom=1in, left=1in, right=1in]{geometry}
\usepackage{xspace}
\usepackage{lineno}
\usepackage{authblk}

\graphicspath{{figures/}}

\newtheorem{theorem}{Theorem}
\newtheorem{corollary}{Corollary}

\newtheorem{proposition}{Proposition}

\newtheorem{remark}[theorem]{Remark}

\newcommand{\drop}[1]{}

\newcommand{\ourProblem}{\textsc{BudgetedDelivery}\xspace}
\newcommand{\returning}{\emph{Returning}\xspace}
\newcommand{\nonreturning}{\emph{Non-Returning}\xspace}

\newcommand{\NP}{$\mathrm{NP}$}
\newcommand{\true}{\mathrm{true}}
\newcommand{\false}{\mathrm{false}}
\newcommand{\bigO}{\mathcal{O}}

\begin{document}
\pagestyle{headings}  

\title{Collaborative Delivery with Energy-Constrained Mobile Robots\footnote{This work was partially supported by the project ANR-ANCOR (anr-14-CE36-0002-01) and the SNF (project 200021L\_156620).}}

\author[1]{Andreas Bärtschi}
\author[2]{Jérémie Chalopin}
\author[2]{Shantanu Das}
\author[3]{Yann Disser}
\author[1]{Barbara Geissmann}
\author[1]{Daniel Graf}
\author[2]{Arnaud Labourel}
\author[4]{Mat\'{u}\v{s} Mihal\'{a}k}
\affil[1]{
	\small ETH Zürich, Switzerland\\ 
	\texttt{$\left\{\right.$baertschi,barbara.geissmann,daniel.graf$\left.\right\}$@inf.ethz.ch}
}
\affil[2]{
	\small LIF, CNRS et Aix-Marseille Université, France\\
  	\texttt{$\left\{\right.$shantanu.das,jeremie.chalopin,arnaud.labourel$\left.\right\}$@lif.univ-mrs.fr}
}
\affil[3]{
	\small Technische Universität Berlin, Germany\\ 
	\texttt{disser@math.tu-berlin.de}
}
\affil[4]{
	\small Maastricht University, Netherlands\\ 
  	\texttt{matus.mihalak@maastrichtuniversity.nl}
}

\date{}
\maketitle

\begin{abstract}
	We consider the problem of collectively delivering some message from a specified source to a designated target location in a graph, using multiple mobile agents. 
	Each agent has a limited energy which constrains the distance it can move. Hence multiple agents need to collaborate to move the message, each agent handing over the message to the next agent to carry it forward. 
	Given the positions of the agents in the graph and their respective budgets, the problem of finding a feasible movement schedule for the agents can be challenging. 
	We consider two variants of the problem: in \emph{non-returning} delivery, the agents can stop anywhere; 
	whereas in \emph{returning} delivery, each agent needs to return to its starting location, a variant which has not been studied before. 
	
	We first provide a polynomial-time algorithm for returning delivery on trees, which is in contrast to the known (weak) \NP-hardness of the non-returning version.
	In addition, we give resource-augmented algorithms for returning delivery in general graphs.
	Finally, we give tight lower bounds on the required resource augmentation for both variants of the problem.
	In this sense, our results close the gap left by previous research.
\end{abstract}

\section{Introduction}
\label{sec:introduction}

We consider a team of mobile robots which are assigned a task that they need to perform collaboratively. Even simple tasks such as collecting information and delivering it to a target location can become challenging when it involves the cooperation of several agents. The difficulty of collaboration can be due to several limitations of the agents, such as limited communication, restricted vision or the lack of persistent memory, and this has been the subject of extensive research (see \cite{MACbook12} for a recent survey). When considering agents that move physically (such as mobile robots or automated vehicles), a major limitation of the agents are their energy resources, which restricts the travel distance of the agent. This is particularly true for small battery operated robots or drones, for which the energy limitation is the real bottleneck. We consider a set of mobile agents where each agent $i$ has a budget $B_i$ on the distance it can move, as in \cite{AnayaCCLPV16,DDalgosensors13}. We model their environment as an undirected edge-weighted graph $G$, with each agent starting on some vertex of $G$ and traveling along edges of $G$, until it runs out of energy and stops forever. In this model, the agents are obliged to collaborate as no single agent can usually perform the required task on its own. 

The problem we consider is that of moving some information from a given source location to a target location in the graph $G$ using a subset of the agents. Although the problem sounds simple, finding a valid schedule for the agents to deliver the message, is computationally hard even if we are given full information on the graph and the location of the agents. Given a graph $G$ with designated source and target vertices, and $k$ agents with given starting locations and energy budgets, the decision problem of whether the agents can collectively deliver a single message from the source to the target node in $G$ is called \ourProblem. 
Chalopin et al.~\cite{DDalgosensors13,DDicalp14} showed that \nonreturning \ourProblem is weakly \NP-hard on paths and strongly \NP-hard on general graphs. 

Unlike previous papers, we also consider a version of the problem where each agent needs to return to its starting location after completing its task. This is a natural assumption, e.g. for robots that need to return to their docking station for maintenance or recharging. We call this variant  \returning \ourProblem.  Surprisingly, this variant of the problem is easier to solve when the graph is a tree (unlike the original version of the problem), but we show it to be strongly \NP-hard even for planar graphs. We present a polynomial time algorithm for solving \returning \ourProblem on trees.

For arbitrary graphs, we are interested in resource-augmented algorithms. Since finding a feasible schedule for \ourProblem is computationally hard when the agents have just enough energy to make delivery possible, we consider augmenting the energy of each robot by a constant factor $\gamma$, to enable a polynomial-time solution to the problem. 
{Given an instance of \ourProblem and some $\gamma>1$, we have a $\gamma$-resource-augmented algorithm, if the algorithm, running in polynomial time, either (correctly) answers that there is no feasible schedule, or finds a feasible schedule for the modified instance with augmented budgets $\hat{B}_i=\gamma \cdot B_i$ for each agent $i$. }

\drop{We provide tight results on resource augmentation of \returning \ourProblem in arbitrary graphs, giving a $2$-resource-augmented algorithm and proving that there exists no ($2-\epsilon$)-resource-augmented algorithm unless P=NP. Similarly for the \nonreturning version of \ourProblem, we show that there exists no ($3-\epsilon$)-resource-augmented algorithm which matches the positive result of~\cite{DDalgosensors13} who presented a $3$-resource-augmented algorithm for the problem. We show that both versions of the problem are strongly \NP-hard even for planar graphs. Finally, we consider a relaxation of the problem where we are given as input an order on the agents and the agents are restricted to move in the prescribed order. In this particular case, the problem is solvable in polynomial time on arbitrary graphs and we provide such a solution.}

\subparagraph{Our Model.} 
We consider an undirected edge-weighted graph $G = (V,E)$ with $n=|V|$ vertices and $m=|E|$ edges.
The weight $w(e)$ of an edge $e \in E$ defines the energy required to cross the edge in either direction.
We have $k$ mobile agents which are initially placed on arbitrary nodes $p_1, \dots, p_k$ of $G$, called starting positions. Each agent $i$ has an initially assigned budget $B_i \in \mathbb{R}_{\geq 0}$ and can move along the edges of the graph, for a total distance of at most $B_i$ (if an agent travels only on a part of an edge,
 its travelled distance is downscaled proportionally to the part travelled).
The agents are required to move a message from a given source node $s$ to a target node $t$. An agent can pick up the message from its current location, carry it to another location (a vertex or a point inside an edge), and drop it there. 
Agents have global knowledge of the graph and may communicate freely.

Given a graph $G$ with vertices $s\neq t \in V(G)$ and the starting nodes and budgets for the $k$ agents, we define \ourProblem as the decision problem of whether the agents can collectively deliver the message without exceeding their individual budgets. 
In \returning \ourProblem each agent needs to return to its respective starting position before using up its energy budget;
in the \nonreturning version we do not place such a restriction on the agents and an agent may terminate at any location in the graph.\\
A solution to \ourProblem is given in the form of a \textit{schedule} which prescribes for each agent whether it moves and if so, the two locations in which it has to pick up and drop off the message.
A schedule is \textit{feasible} if the message can be delivered from $s$ to $t$.

\subparagraph{Related Work.} 

Delivery problems in the graph have been usually studied for a single agent moving in the graph. For example, the well known \emph{Travelling salesman problem} (TSP) or and the \emph{Chinese postman problem} (CPP) require an agent to deliver packets to multiple destinations located in the nodes of the graph or the edges of the graph. The optimization problem of minimizing the total distance traveled is known to be \NP-hard~\cite{ApplegateTSP} for TSP, but can be solved in polynomial time for the CPP~\cite{Edmonds73}. 

When the graph is not known in advance, the problem of exploring a graph by a single agent has been studied with the objective of minimizing the number of edges traversed (see e.g. ~\cite{PanPe99,AlbH00}). Exploration by a team of two agents that can communicate at a distance,  has been studied by Bender and Slonim~\cite{BS94} for digraphs without node identifiers. 
The model of energy-constrained robot was introduced by Betke et al.~\cite{Betke95} for single agent exploration of grid graphs. Later Awerbuch et al.~\cite{AwBS99} studied the same problem for general graphs. In both these papers, the agent could return to its starting node to refuel and between two visits to the starting node, the agent could traverse at most $B$ edges. Duncan et al.~\cite{Duncan01} studied a similar model where the agent is tied with a rope of length $B$ to the starting location and they optimized the exploration time, giving an $\bigO(m)$ time algorithm.  

For energy-constrained agents without the option of refuelling, multiple agents may be needed to explore even graphs of restricted diameter. Given a graph $G$ and $k$ agents starting from the same location, each having an energy constraint of $B$, deciding whether $G$ can be explored by the agents is \NP-hard, even if graph $G$ is a tree~\cite{FraGKP06}. Dynia et al.~studied the online version of the problem~\cite{poweraware06,Dynia2007}. They presented algorithms for exploration of trees by $k$ agents when the energy of each agent is augmented by a constant factor over the minimum energy $B$ required per agent in the offline solution. Das et al.~\cite{DDK2015} presented online algorithms that optimize the number of agents used for tree exploration when each agent has a fixed energy bound $B$. On the other hand, Dereniowski et al.~\cite{Dereniowski2015} gave an optimal time algorithm for exploring general graphs using a large number of agents.  
Ortolf et al.~\cite{Ortolf2012} showed bounds on the competitive ratio of online exploration of grid graphs with obstacles, using $k$ agents.

When multiple agents start from arbitrary locations in a graph, optimizing the total energy consumption of the agents is computationally hard for several formation problems which require the agents to place themselves in desired configurations (e.g. connected or independent configurations) in a graph. Demaine et al.~\cite{Demaine2009} studied such optimization problems and provided approximation algorithms and inapproximability results. Similar problems have been studied for agents moving in the visibility graphs of simple polygons and optimizing either the total energy consumed or the maximum energy consumed per agent can be hard to approximate even in this setting, as shown by Bilo et al.~\cite{Bilo2013}. 

Anaya et al.~\cite{AnayaCCLPV16} studied centralized and distributed algorithms for the information exchange by energy-constrained agents, in particular the problem of transferring information from one agent to all others (\emph{Broadcast}) and from all agents to one agent (\emph{Convergecast}). For both problems, they provided hardness results for trees and approximation algorithms for arbitrary graphs. The budgeted delivery problem was studied by Chalopin et al.~\cite{DDalgosensors13} who presented hardness results for general graphs as well as resource-augmented algorithms. For the simpler case of lines,  \cite{DDicalp14} proved that the problem is weakly \NP-hard and presented a quasi-pseudo-polynomial time algorithm. 
Czyzowicz et al.~\cite{EnergyExchange15} recently showed that the problems of budgeted delivery, broadcast and convergecast remain \NP-hard for general graphs even if the agents are allowed to exchange energy when they meet.

\subparagraph{Our Contribution.} 

This is the first paper to study the \returning version of  \ourProblem. We first show that this problem can be solved in $\bigO(n + k \log{k})$ time for lines and trees (Section~\ref{sec:line}). This is in sharp contrast to the \nonreturning version which was shown to be weakly \NP-hard~\cite{DDicalp14} even on lines. In Section~\ref{sec:planar-hardness}, we prove that \returning \ourProblem is \NP-hard even for planar graphs. For arbitrary graphs with arbitrary values of agent budgets, we present a $2$-resource-augmented algorithm and we prove that this is the best possible, as there exists no ($2-\epsilon$)-resource-augmented algorithm unless $\mathrm{P}=\mathrm{NP}$ (Section~\ref{sec:resource-augmentation-hardness}). We show that this bound can be broken when the agents have the same energy budget and we present a $(2-2/k)$-resource-augmented algorithm for this case.

For the \nonreturning version of the \ourProblem, we close the gaps left open by previous research~\cite{DDalgosensors13,DDicalp14}. In particular we prove that this variant of the problem is also strongly \NP-hard on planar graphs, while it was known to be strongly \NP-hard for general graphs and weakly \NP-hard on trees. We also show tightness of the $3$-resource-augmented algorithm for the problem, presented in \cite{DDalgosensors13}. Finally, in Section~\ref{sec:discussion}, we investigate the source of hardness for \ourProblem and show that the problem becomes easy when the order in which the agents pick up the message is known in advance. 
\drop{For this simplified case, we present a polynomial time solution for \returning \ourProblem.}

\section{Returning BudgetedDelivery on the Tree}
\label{sec:line}
We study the \returning \ourProblem on a tree and show that it can be solved in polynomial time. We immediately observe that this problem is reducible to the \returning \ourProblem on a path: There is a unique s-t path on a tree and we can move each agent from her starting position to the nearest node on this s-t path while subtracting from her budget twice the distance traveled.
The path problem now has an equivalent geometric representation on the line: the source node $s$, the target node $t$, and the starting positions of the agents $p_i$ are coordinates of the real line. We assume $s<t$, i.e., the message needs to be delivered from left to right.

Without loss of generality, we consider schedules in which every agent $i$ that moves uses all its budget $B_i$.
Because every agent needs to return to its starting position, an agent $i$ can carry the message 
on any interval of size $B_i/2$ that contains the starting position $p_i$.
For every agent $i$, let $l_i = p_i - B_i/2$ denote the leftmost point where she can pick a message, and let $r_i = p_i + B_i/2$ be the rightmost point to where she can deliver the message.
The \returning \ourProblem on a line now becomes the following \emph{covering problem}: Can we choose, for every $i$, an interval $I_i$ of size $B_i/2$ that lies completely within the \emph{region} $[l_i,r_i]$ such that the segment $[s,t]$ is covered by the chosen intervals, i.e., such that $[s,t]\subseteq \cup_i I_i$?

The following \textit{greedy} algorithm solves the covering problem. The algorithm works iteratively in rounds $r=1,2,\ldots$. We initially set $s_1=s$. We stop the algorithm whenever $s_r\geq t$, and return $\true$.
In round $r$, we pick $i^*$ having the smallest $r_{i^*}$ among all agents $i$ with $l_i \leq s_r < r_i$, and set $s_{r+1} = \min\{r_{i^*},s_r + B_{i^*}/2\}$ and $I_{i^*} = (s_{r+1} - B_{i^*}/2, s_{r+1})$, and continue with the next round $r+1$. If we cannot choose $i^*$, we stop the algorithm and return $\false$.

\begin{theorem}
There is an $\bigO(n + k\log k)$-time algorithm for \returning \ourProblem on a tree.
\end{theorem}

\begin{proof}
The reduction from a tree to a path takes $\bigO(n)$ time using breadth-first search from $s$ and the algorithm \textit{greedy} can be implemented in time $\bigO(k\log k)$ using a priority queue.

For the correctness, we now show that greedy returns a solution to the covering problem if and only if there exists one.
Greedy can be seen as advancing the cover of $[s,t]$ from left to right by adding intervals $I_i$. 
Whenever it decides upon $I_i$, it will set $s_r$ to the respective endpoint of $I_i$, and never ever consider $i$ again or change the placement of $I_i$ within the boundaries $[l_i,r_i]$.
Thus, whenever $s_r\geq t$, the intervals $I_i$ form a cover of $[s,t]$.

We now show that if a cover exists, greedy finds one.
Observe first that a cover can be given by a subset of the agents $\{i_1,\ldots,i_t\}$, $t\leq k$, and by their ordering $(i_1,i_2,\ldots)$, according to the right endpoints of their intervals $I_{i_j}$, since we can reconstruct a covering by always placing the respective interval $I_{i_j}$ at the rightmost possible position.

Suppose, for contradiction, that greedy fails. 
Let $(i_1^*,i_2^*,\ldots)$ be a minimal cover of $[s,t]$ that agrees with the greedy schedule $(i_1,i_2,\ldots)$ in the maximum number of first agents $i_1,\ldots,i_j$.
Hence, $j+1$ is the first position such that $i^*_{j+1}\neq i_{j+1}$.
The left endpoints of $I^*_{j+1}$ and $I_{j+1}$ correspond to $s_{r+1}$ in our algorithm. 
If agent $i_{j+1}$ does not appear in the solution $(i_1^*,i_2^*,\ldots)$, adding $i_{j+1}$ to that solution and deleting some of the subsequent ones results in a minimal cover that agrees on the first $j+1$ agents, a contradiction.

If agent $i_{j+1}$ appears in the solution $(i_1^*,i_2^*,\ldots)$, say, as agent $i_{j+\delta}^*$, then we modify this cover by swapping $i_{j+1}^*$ with $i_{j+\delta}^*$.
We claim that the new solution still covers $[s,t]$.
This follows immediately by observing that greedy chose $i_{j+1}$ to have smallest $r_i$ among all agents that can extend the covering beyond  $s_{r+1}$. Since every agent $i$ covers at least half of its region $[l_i,r_i]$, we know that $i_{j+1}^*$ and $i_{j+\delta}^*$ together cover the region $[s_{r+1},r_{i_{j+\delta}^*}]$, and therefore by minimality $i_{j+\delta}^* = i_{j+2}^*$. Finally, if we change the order of the two agents, they will still cover the region $[s_{r+1},r_{i_{j+\delta}^*}]$ (see Figure~\ref{fig:swap}).
\begin{figure}[t]
	\includegraphics[width=\textwidth]{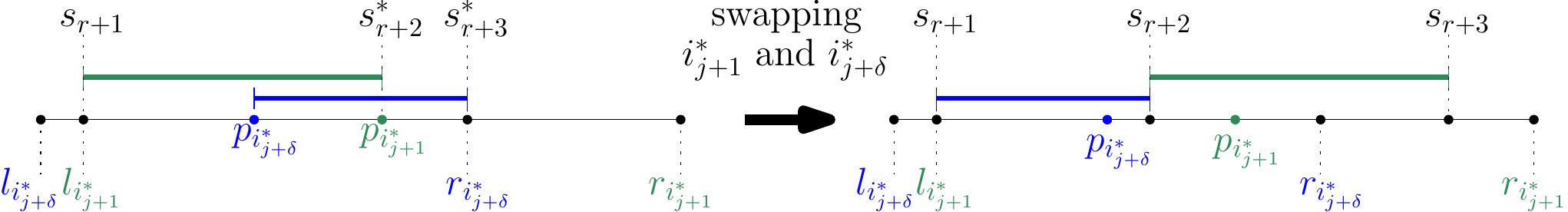}
	\caption{Changing the order of agents $i_{j+1}^*$ and $i_{j+\delta}^*$ in the schedule.\label{fig:swap}}
\end{figure}
\end{proof}

\section{Resource Augmentation Algorithms}
\label{sec:resource-augmentation}
We now look at general graphs $G=(V,E)$. As we will see in the next section, \ourProblem is \NP-hard, hence we augment the budget of each agent by a factor $\gamma>1$ to allow for polynomial-time solutions.
For non-returning agents, a $\min\left\{ 3, 1+\max\smash{\tfrac{B_i}{B_j}} \vphantom{3^2}\right\}$-resource-augmented algorithm was given by Chalopin~et.~al.~\cite{DDalgosensors13}.
We first provide a $2$-resource-augmented algorithm for \returning \ourProblem. This is tight as there is no polynomial-time $(2-\varepsilon)$-resource-augmented algorithm, 
unless $\mathrm{P}=\mathrm{NP}$ (Section~\ref{sec:resource-augmentation-hardness}). If, however, the budgets of the agents are similar, we can go below the $2$-barrier: 
In this case, we present a $\left(1+\smash{\tfrac{k-2}{k}\max\tfrac{B_i}{B_j}} \vphantom{3^2} \right)$-resource-augmented algorithm.
Throughout this section, we assume that there is no feasible schedule with a single agent, which we can easily verify.

\subparagraph{Preliminaries.} We denote by $d(u,v)$ the distance of two points $u,v \in G$. Assume an agent $i$ with budget $B_i$ starts in $u$ and moves first to $v$. 
Which locations in the graph (vertices and positions on the edges) are still reachable by $i$ so that he has sufficent energy left to move back to $u$? We define the ellipsoid 
$\mathcal{E}(u,v,B_i) = \left\{ p \in G \ | \ d(u,v)+d(v,p)+d(p,u) \leq B_i \right\}$
and the ball $\mathcal{B}(u,\tfrac{B_i}{2}) = \mathcal{E}(u,u,B_i)$.
It is easy to see that $\mathcal{E}(u,v,B_i)$ can be 
(i) computed in polynomial time by running Dijkstra's shortest path algorithm from both $u$ and $v$ and (ii) represented in linear space: 
We store all vertices $p \in V$ with $p\in \mathcal{E}(u,v,B_i)$, and for each edge $(p,q) \in E$ with $p\in \mathcal{E}(u,v,B_i), q\notin \mathcal{E}(u,v,B_i)$ we store the furthest point of $(p,q)$ still reachable by $i$.

\begin{theorem}[2-resource-augmentation] 
	There is a polynomial-time $2$-resource-augmented algorithm for \returning \ourProblem.
	\label{thm:2-resource-augmentation}
\end{theorem}
\begin{proof}
	Denote by $p_i$ the starting position of agent $i$. We consider the balls $\mathcal{B}_i := \mathcal{B}(p_i, \tfrac{B_i}{2})$ around all agents,
	as well as the balls $\mathcal{B}(s,0)$ and $\mathcal{B}(t,0)$ of radius 0 around $s$ and $t$. We compute the \emph{intersection graph} $G_I$ of the balls, which can be done in polynomial time. 
	If there is a feasible schedule, then there must be a path from $\mathcal{B}(s,0)$ to $\mathcal{B}(t,0)$ in $G_I$ (for example the path given by the balls around the agents in the feasible schedule).

	If there is no path from $\mathcal{B}(s,0)$ to $\mathcal{B}(t,0)$, then the algorithm outputs that there is no feasible schedule with non-augmented budgets. 
	Otherwise we can get a 2-resource-augmentation as follows:
	Pick a shortest path from $\mathcal{B}(s,0)$ to $\mathcal{B}(t,0)$ in $G_I$ and denote by $\ell \leq k$ the number of agents on this path, labeled without loss of generality $1,2,\ldots, \ell$.
	For each edge on the shortest path, we specify a handover point $h_i \in \mathcal{B}_i \cap \mathcal{B}_{i+1}$ in $G$ (where we set $h_{0}=s$ and $h_{\ell} = t$). 
	Then each agent $i,\ i = 1,\ldots,\ell$ walks from its starting position $p_i$ to the handover point $h_{i-1}$ to pick up the message,
	goes on to the handover point $h_{i}$ to drop the message there, and returns home to $p_i$. Since $h_{i-1}, h_{i} \in \mathcal{B}(p_i,\tfrac{B_i}{2})$, 
	the budget needed by agent $i$ to do so is at most $d(p_i, h_{i-1}) + d(h_{i-1}, h_{i}) + d(h_{i},p_i) \leq \tfrac{B_i}{2} + 2\cdot \tfrac{B_i}{2} + \tfrac{B_i}{2} = 2\cdot B_i$.	
\end{proof}

\begin{theorem}
	There is a polynomial-time $\left(1+ \smash{\tfrac{k-2}{k}\max \tfrac{B_j}{B_i}} \vphantom{3^2} \right)$-resource-augmented algorithm for \returning \ourProblem. 
	\label{thm:advanced-resource-augmentation}
\end{theorem}
\begin{proof} 
	We first ``guess'' the first agent $a$ and the last agent $b$ of the feasible schedule (by trying all ${k \choose 2}$ pairs). 
	In contrast to Theorem~\ref{thm:2-resource-augmentation}, we can in this way get a $2$-resource-augmented solution in which $a$ and $b$ only need their original budgets.
	Intuitively, we can evenly redistribute the remaining part of $\hat{B}_a$ and $\hat{B}_b$ among all $k$ agents, such that for each agent $i$ we have $\hat{B}_i \leq B_i+\tfrac{k-2}{k} \max B_j$.
	Without loss of generality, we assume that agent $a$ walks from its starting position on a shortest path to $s$ to pick up the message, 
	and that agent $b$ walks home directly after dropping the message at $t$. 
	Hence consider the ellipsoids $\mathcal{B}_a := \mathcal{E}(p_{a},s,B_{a})$ and $\mathcal{B}_b := \mathcal{E}(p_{b},s,B_{b})$ as well as the balls $\mathcal{B}_i := \mathcal{B}(p_i, \tfrac{B_i}{2})$ around the starting positions of all other agents
	and compute their intersection graph $G_I$.

	We denote by $i = 1,\ldots,\ell$ the agents on a shortest path from $\mathcal{B}_a$ to $\mathcal{B}_b$ in $G_I$ (if any), where $a=1,\ b=\ell\leq k$ and
	specify the following points: $h_0 = s$, $h_i \in \mathcal{B}_i \cap \mathcal{B}_{i+1}$, and $h_{\ell} = t$.
	If the agents handover the message at the locations $h_i$, we get a $2$-resource-augmentation where the agents $1$ and $\ell$ use only their original budget. 
	Instead we let them help their neighbours $2$ and $\ell-1$ by $\tfrac{\ell-2}{\ell}B_2$ and $\tfrac{\ell-2}{\ell}B_{\ell-1}$, respectively. 
	Those agents further propagate the surplus towards the agent(s) in the middle,
	see Figure~\ref{fig:resource-augmentation-example} (right).
	
	\begin{figure}[ht!]
		\centering
		\includegraphics[width=\linewidth]{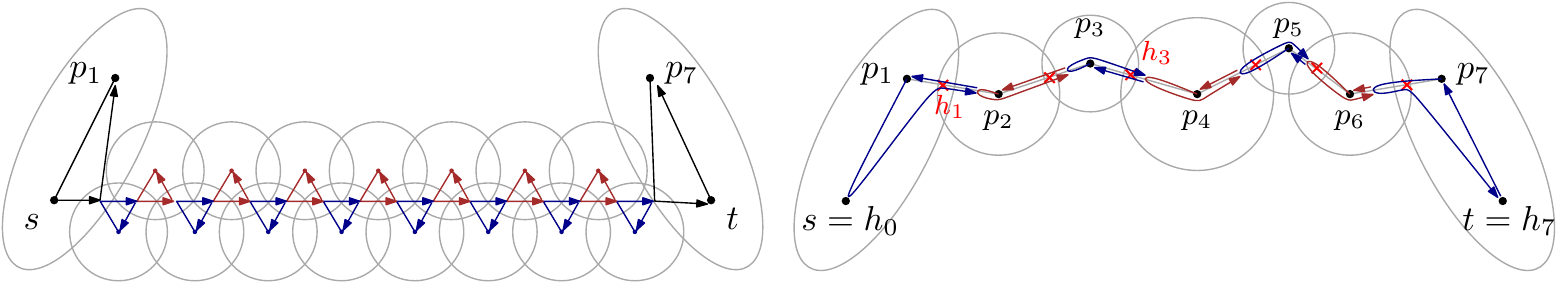}
		\caption{(left) Feasible schedule. (right) Schedule with $\left(1+ \tfrac{5}{7}\max \tfrac{B_j}{B_i} \right)$-resource-augmentation.}
		\label{fig:resource-augmentation-example}
	\end{figure}
	\noindent Specifically, we let the agents move as follows:
	
	Agent $1$ goes to $h_0$ to pick up the message and then goes on to $h_1$. Then he moves towards $p_2$ along the shortest path from $h_1$ to $p_2$ by a $\tfrac{\ell-2}{ell}$-fraction of $d(h_1,p_2)$, 
	drops off the message and returns home. The budget needed to do so is at most $d(p_1,h_0)+d(h_0,h_1)+2\tfrac{\ell-2}{\ell}d(h_1,p_2)+d(h_1,p_1) \leq B_1 + \tfrac{\ell-2}{\ell}B_2$.
	Agents $i=2,\ldots,\lfloor\tfrac{\ell}{2}\rfloor$ get help from their preceding agent and thus can help the following agent:
	Agent $i$ walks from its starting position $p_i$ by a $\tfrac{2(i-1)}{\ell}$-fraction towards $h_{i-1}$ to pick up the message and then returns home.
	Then $i$ goes on to the point $h_i$ and from there on by a $\tfrac{\ell-2i}{\ell}$-fraction towards $p_{i+1}$ to drop off the message. Finally, agent $i$ returns home to $p_i$. 
	Since $h_{i-1}, h_{i} \in \mathcal{B}_i$ and $h_i \in \mathcal{B}_{i+1}$, the budget needed by agent $i$ to do so is at most
	$2\tfrac{2(i-1)}{\ell} d(p_i, h_{i-1}) + 2d(p_i, h_i) + 2\tfrac{\ell-2i}{\ell}d(h_i,p_{i+1}) \leq $
	$\tfrac{2(i-1)}{\ell}B_i + B_i + \tfrac{\ell-2i}{\ell}B_{i+1} \leq B_i + \tfrac{\ell-2}{\ell}\max\left\{ B_i, B_{i+1} \right\}$.
	Agents $i=\lceil\tfrac{\ell+2}{2}\rceil, \ldots, \ell$ help in the same way their preceding agent, hence they need a budget of at most $B_i + \tfrac{\ell-2}{\ell}\max\left\{ B_{i-1}, B_i \right\}$. 
	If $\ell$ is odd there is an additional middle agent $i=\tfrac{\ell+1}{2}$ who needs a budget of at most $2\tfrac{\ell-1}{\ell}d(p_i,h_{i-1}) + 2\tfrac{\ell-1}{\ell}d(p_i,h_{i+1}) \leq 1+\tfrac{\ell-2}{\ell}B_i$.
	
	Hence we achieve a resource augmentation of $1+ \smash{\tfrac{\ell-2}{\ell}\max \tfrac{B_j}{B_i}} \leq 1+ \smash{\tfrac{k-2}{k}\max \tfrac{B_j}{B_i}}$.
\end{proof}

\section{Hardness for Planar Graphs}
\label{sec:planar-hardness}
In this section, we show that \ourProblem in a planar graph is strongly \NP-hard, both for the \returning version and the \nonreturning version. 
Both proofs are based on the same reduction from \textsc{Planar3SAT}.
%
\subparagraph{Planar 3SAT.} Let $F$ be a conjunctive normal form 3CNF with a set of variables $V=\left\{ v_1, \dots, v_x \right\}$ and a set of clauses $C=\left\{ c_1, \dots, c_y \right\}$.
Each clause is a disjunction of at most three literals $\ell(v_i) \vee \ell(v_j) \vee \ell(v_k)$, where $\ell(v_i) \in \left\{ v_i, \overline{v_i} \right\}$.
We can represent $F$ by a graph $H(F) = (B \cup V,A_1 \cup A_2)$ which we build as follows: 
We start with a bipartite graph with the node set $N$ consisting of all clauses and all variables and an edge set $A_1$ which contains an edge between each clause $c$ and variable $v$ if and only if $v$ or $\overline{v}$ is contained in $c$, $A_1 = \left\{ \left\{ c_i, v_j \right\} \ | \ v_j \in c_i \text{ or } \overline{v_j} \in c_i \right\}$.
To this graph we add a cycle $A_2$ consisting of edges between all pairs of consecutive variables, $A_2 = \left\{ \left\{ v_j, v_{j+1} \right\} \ | \ 1\leq j < x \right\} \cup \left\{ v_x, v_1 \right\}.$
We call $F$ \emph{planar} if there is a plane embedding of $H(F)$ which \emph{at each variable node} has all paths representing positive literals on one side of the cycle $A_2$ and all paths representing negative literals on the other side of $A_2$.
The decision problem \textsc{Planar3SAT} of finding whether a given planar 3CNF $F$ is satisfiable or not is \NP-complete, a result due to Lichtenstein~\cite{planar3sat82}.
We assume without loss of generality that every clause contains at most one literal per variable.
For an example of such an embedding, see Figure~\ref{fig:planar3sat} (left).

\subparagraph{Building the Delivery Graph.} 
We first describe how to turn a plane embedding of a planar 3CNF graph $H(F)$ into a delivery graph $G(F)$, see Figure~\ref{fig:planar3sat}. 
Only later we will define edge weights, the agents' starting positions and their energy budgets.
We will focus on \returning \ourProblem; the only difference for non-returning agents lie in their budgets, we provide adapted values for non-returning agents in footnotes.

\begin{figure}[t!]
	\centering
	\includegraphics[width=\linewidth]{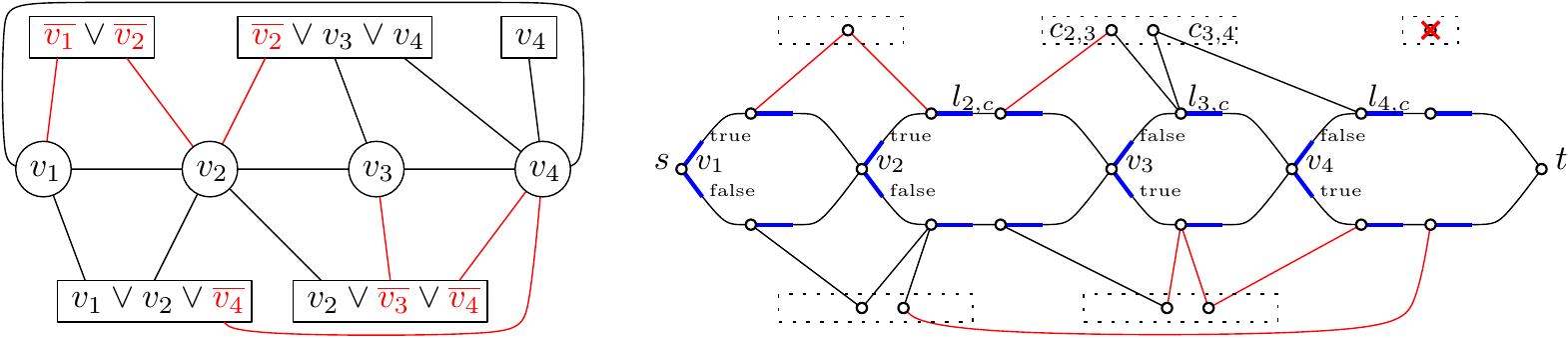}
	\caption{(left) A plane embedding of a 3CNF $F$ which is satisfied by
		$(v_1,v_2,v_3,v_4) = (\true, \false, \false, \true)$.
		(right) Its transformation to the corresponding delivery graph.}
	\label{fig:planar3sat}
\end{figure}

We transform the graph in four sequential steps: 
First we dissolve the edge $\left\{ v_x, v_1 \right\}$ and replace it by an edge $\left\{ v_x, v_{x+1} \right\}$. 
Secondly, denote by $\deg_{H(F),A_1}(v)$ the total number of positive literal edges and negative literal edges adjacent to $v$. 
Then we can ``disconnect'' and ``reconnect'' each variable node $v_i$ ($1\leq i \leq n$) from all of its adjacent clause nodes as follows:
We delete all edges $\left\{ \left\{ v_i, c \right\} \right\} \subseteq A_1$ and split $\left\{ v_i, v_{i+1} \right\}$ into two paths $p_{i,\true}$ and $p_{i,\false}$, on which we place a total of $\deg_{A_1}(v)$ internal \emph{literal nodes} $\mathit{l}_{i,c}$: 
If $v_i$ is contained in a clause $c$ -- and thus we previously deleted $\left\{ v_i, c \right\}$ -- we place $\mathit{l}_{i,c}$ on $p_{i,\false}$ and ``reconnect'' the variable by adding an edge between $\mathit{l}_{i,c}$ and the clause node $c$.
Else if $\overline{v}$ is contained in $c$ we proceed similarly (putting the node $\mathit{l}_{i,c}$ on $p_{i,\true}$ instead).
As a third step, depending on the number of literals of each clause $c$, we may modify its node:
If $c$ contains only a single literal, we delete the $c$ node. 
If $c$ contains two literals $\ell(v_i), \ell(v_j)$, we rename the node to $c_{i,j}$. 
If $c$ is a disjunction of three literals $\ell(v_i), \ell(v_j), \ell(v_k)$, we split it into two nodes $c_{i,j}$ (connected to $l_{i,c}$, $l_{j,c}$) and $c_{j,k}$ (connected to $l_{j,c}$, $l_{k,c}$).
Finally, we place the message on the first variable node $s:= v_1$ and set its destination to $t:=v_{x+1}$. 

We remark that all four steps can be implemented such that the resulting delivery graph $G(F)$ is still planar, as illustrated in Figure~\ref{fig:planar3sat}
(in each path tuple $(p_{i,\true},\ p_{i,\false})$ the order of the internal nodes follows the original circular order of adjacent edges of $v_i$, 
and for each clause $c = \ell(v_i) \vee \ell(v_j) \vee \ell(v_k) $ the nodes $c_{i,j}$ and $c_{j,k}$ are placed close to each other). 

\subparagraph{Reduction Idea.} We show that the message can't be delivered via any of the clause nodes. 
Thus the message has to be routed in each path pair $(p_{i,\true},\ p_{i,\false})$ through exactly one of the two paths.
If the message is routed via the path $p_{i,\true}$, we interpret this as setting $v_i = \true$ and hence we can read from the message trajectory a satisfiable assignment for $F$. 

\subparagraph{Agent Placement and Budgets.} We will use greek letters for weights (namely $\zeta$ and $\delta$) when the weights depend on each other or
on the input. We place three kinds of agents on $G$:
\begin{enumerate}
	\item	\emph{Variable agents:} $x$ agents which are assigned to the variable nodes $v_1, \ldots, v_x$. 
		These agents will have to decide whether the message is delivered via $p_{i,\true}$ or via $p_{i,\false}$, thus setting the corresponding variable to true or to false. 
		We give all of them a budget of $2\zeta$.\footnote{In the \nonreturning version we want agents to have the same ``range'', hence we set their budget to $\zeta$.}
	\item	\emph{Clause agents:} One agent per created clause \emph{node}, e.g. a clause $c$ containing three literals gets two agents, one in each of the two clause nodes. 
		We think of these agents as follows: If in $c = \ell(v_i) \vee \ell(v_j) \vee \ell(v_k) $ the literal $\ell(v_j)$ is false, 
		then clause $c$ needs to send one of its agents down to the corresponding path node $l_{j,c}$ to help transporting the message over the adjacent ``gap'' of size $\zeta$ 
		(depicted blue in Figures~\ref{fig:planar3sat} (right),~\ref{fig:placement_weights}).
		A 3CNF $F$ will be satisfiable, if and only if no clause needs to spend more agents than are actually assigned to it respectively its node(s) in $G(F)$.
		We give all clause agents a budget of $2\cdot(1+\zeta)$.\footnote{In the \nonreturning version we assign a budget of $(1+\zeta)$ to clause agents.}
	\item	\emph{Separating agents:} These will be placed in-between the agents defined above, to ensure that the variable and clause agents
		actually need to solve the task intended for them (they should not be able to deviate and help out somewhere else -- not even their
		own kind). The separating agents will be placed in pairs inside \emph{$\delta$-tubes}, which we define below. 
\end{enumerate}

\begin{remark}
Strictly speaking, a reduction without variable agents works as well. In terms of clarity, we like to think of variable agents as the ones setting the variables to $\true$ or $\false$.
\end{remark}
\subparagraph{$\delta$-Tubes.} We call a line segment a \emph{$\delta$-tube} if it satisfies the following four properties: (i) It has a length of
$\delta$. (ii) It contains exactly two agents which both have budget at most $\delta$. (iii) Neither agent has enough energy to leave the line segment
on the left or on the right by more than a distance of $\frac{\delta}{3}$. (iv) The agents can collectively transport a message through the line
segment from left to right.\\
$\delta$-tubes exist for both \ourProblem versions, examples are given in Figure~\ref{fig:placement_weights}~(left). 
The reader may think of these examples, whenever we talk about $\delta$-tubes.
\subparagraph{Edge Weights.} We define edge weights on our graph $G(F)$ as follows: All edges between clause nodes and internal path nodes get weight~$1$ (in particular
this means that if a clause agent walks to the path, it has a remaining range of $\zeta$). Each path consists of alternating pieces of length $\zeta$
and of $\delta$-tubes. We choose $\delta := \tfrac{4\zeta}{3} > \zeta$. This means that neither variable nor clause agents can cross a $\delta$-tube
(because their budget is not sufficiently large). Furthermore the distance any separating agent can move outside of its residential tube is at most
$\tfrac{\delta}{3} = \tfrac{4\zeta}{9} < \tfrac{\zeta}{2}$. In particular separating agents are not able to collectively transport the message over a
$\zeta$-segment, since from both sides they are not able to reach the middle of the segment to handover the message. At last we set $\zeta :=
\tfrac{1}{8}$.

\begin{figure}[t!]
	\centering
	\includegraphics[width=\linewidth]{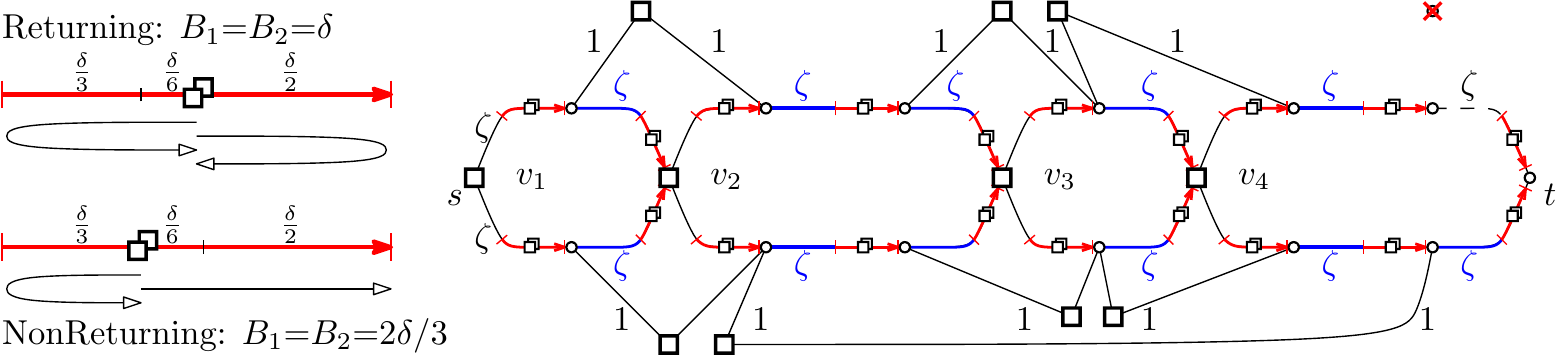}
	\caption{(left) Two examples of $\delta$-tubes for both versions of \ourProblem.\newline
		(right) Agent placement and edge weights on $G(F)$; agents are depicted by squares.}
	\label{fig:placement_weights}
\end{figure}

\begin{restatable}[Returning BudgetedDelivery]{lemma}{lemnphardness}
	A planar 3CNF $F$ is satisfiable if and only if it is possible to deliver a message from $s$ to $t$ in the corresponding
	delivery graph $G(F)$, such that all agents are still able to return to their starting points in the end.
	\label{lem:budget-with-return}
\end{restatable}

\begin{proof}[Proof (Sketch)]
	See Appendix~\ref{app:np-hardness} for a detailed proof.

	``$\Rightarrow$'' The schedule is straightforward: Each \emph{variable agent} chooses, according to the assignment to $v_i$, the $\true$-path $p_{i,\true}$ or the $\false$-path $p_{i,\false}$.
	\emph{Separating agents} and \emph{clause agents} help wherever they are needed.

	``$\Leftarrow$'' One can show that the message cannot be delivered via any clause node. Hence we set $v_i = \true$ if and only if the message moves through $p_{i,\true}$.
	Now, each clause must have one satisfied literal, otherwise its agents could not have helped to bridge all $\zeta$-segments.

\end{proof}

\noindent The same arguments work for \nonreturning \ourProblem as well.
Recall that a delivery graph $G(F)$ created from a planar 3CNF $F$ is planar. Furthermore the size of $G(F)$, as well as the number of agents we use, is
polynomial in the number of clauses and variables. The agents' budgets and the edge weights are polynomial in $\zeta, \delta$ and thus constant.
Thus Lemma~\ref{lem:budget-with-return} shows \NP-hardness of \ourProblem on planar graphs. Finally, note that hardness also holds for a \emph{uniform} budget $B$: 
One can simply add an edge of length $(B-B_i)/2$ to the starting location of each agent $i$ and relocate $i$ to the end of this edge.\footnote{We relocate a non-returning agent by adding an edge of length $(B-B_i)$.}
\begin{theorem}[Hardness of BudgetedDelivery]
	Both versions of \ourProblem are strongly \NP-hard on planar graphs, even for uniform budgets.
	\label{thm:budget-hardness}
\end{theorem}

\section{Hardness of Resource Augmentation}
\label{sec:resource-augmentation-hardness}
\subparagraph{Main Ideas.} We show that for all $\varepsilon>0$, there is no polynomial-time $(2-\varepsilon)$-resource-augmented algorithm for \returning \ourProblem, unless $\mathrm{P}=\mathrm{NP}$. 
The same holds for $(3-\varepsilon)$-resource-augmentation for the \nonreturning version.  
Intuitively, an algorithm which finds out how to deliver the message with resource-augmented agents will at the same time solve 3SAT. We start by taking the reduction
from \textsc{Planar3SAT} from Section~\ref{sec:planar-hardness}. However, in addition to the previous delivery graph construction $G(F)$,
we need to replace the $\delta$-tubes and $\zeta$-segments in order to take care of three potential pitfalls. We illustrate the modification into the
new graph $G'(F)$ in Figure~\ref{fig:replacements}:
\begin{enumerate} 
	\item	In a resource-augmented setting, $\delta$-tubes are no longer able to separate the clause and variable agents: 
		These agents might be able to cross the $\delta$-tube, or the separating agents residing inside the $\delta$-tube can help out in the $\zeta$-segments
		(there is no value for $\delta$ to prevent both). We will tackle this issue below by replacing $\delta$-tubes by a chain of logarithmically
		many tubes with exponentially increasing and decreasing $\delta$-values.
	\item	In the reduction for the original decision version of \ourProblem, a clause $c$ with three literals gave rise to two clause nodes $c_{i,j},\, c_{j,k}$ that were adjacent to the same path node $l_{j,c}$. 
		Hence the agent on $c_{i,j}$, now with resource-augmented budget, could pick up the message at $l_{j,c}$ and bring it close to the second resource-augmented agent stationed at $c_{j,k}$. 
		This agent then might transport the message via its own clause node to the distant literal node $l_{k,c}$. 
		To avoid this, we replace every $\zeta$-segment adjacent to such a ``doubly'' reachable path node $l_{j,c}$ by two small parallel arcs. 
		Both arcs contain exactly one $\zeta$-segment, reachable from only one clause node (the message can then go over either arc),
		as well as a chain of tubes to provide the necessary separation.
	\item	A single clause agent stationed at $c_{i,j}$ might retrieve the message from the first literal node $l_{i,c}$, walk back to its origin and then on to the second literal $l_{j,c}$, 
		thus transporting the message over a clause node. This can always be done by 2-resource-augmented agents; however for
		$(2-\varepsilon)$-resource-augmentation we can prevent this by carefully tuning the weights of the $\zeta$-segments, e.g. such that
		$(2-\varepsilon)\cdot (1+\zeta) \ll 2$.\footnote{Non-returning clause agents can do this if they are $3$-resource-augmented; and we can
		prevent it for $(3-\varepsilon)$-resource-augmentation by setting $\zeta$ such that $(3-\varepsilon)\cdot (1+\zeta) \ll 3$ (in fact
		the value of $\zeta$ will be the same as for \returning \ourProblem, but we will use different bounds in the proof).}
\end{enumerate}
We now give a more formal description of the ideas mentioned above. Recall that a $\delta$-tube had length $\delta$ and contained two agents with
budget at most $\delta$ each. If these agents are now $\gamma$-resource-augmented, $\gamma < 3$, they can move strictly less than $3\delta$ to the right
or to the left of the $\delta$-tube. In the following we want to uncouple the length of the line segment from the range the agents have left to move
on the outside of the line segment.

\begin{figure}[t!]
	\centering
	\includegraphics[width=\linewidth]{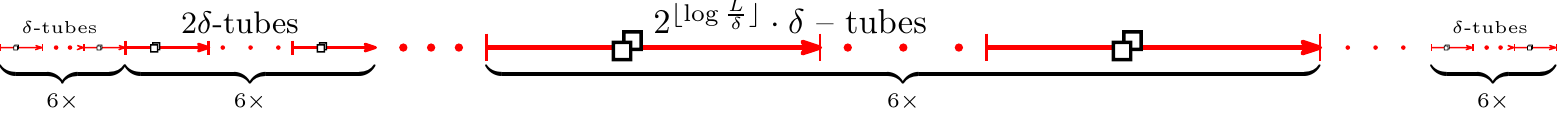}
	\caption{$L$-$\delta$-chains consist of blocks of 6 tubes of exponentially increasing and decreasing size.}
	\label{fig:x-delta-chain}
\end{figure}
\subparagraph{$L$-$\delta$-Chains.} We call a line segment an \emph{$L$-$\delta$-chain} if it satisfies the following three properties: (i) Its length is
at least $L$ (a constant). (ii) No  $\gamma$-resource-augmented agent ($1 \leq \gamma < 3$) contained in the chain has enough energy to leave the line
segment by $3\delta$ or more. (iii) The agents contained in the chain can collectively transport a message through the line segment from left to right
(already with their original budget).\\
We can create $L$-$\delta$-chains for both \ourProblem versions simply by using the respective \emph{$\delta$-tubes} as a blackbox: 
We start our line segment by adding a block of six $\delta$-tubes next to each other, followed by a block of six $2\delta$-tubes, a block of six $4\delta$-tubes and so on
until we get a block of length at least $6\cdot 2^{\lfloor \log L/\delta \rfloor}\cdot \delta > L$. The same way we continue to add
blocks of six tubes with lengths decreasing by powers of 2, see Figure~\ref{fig:x-delta-chain}. Obviously properties (i) and (iii) are satisfied. To
see (ii), note that any agent contained in the first or last block of $\delta$-tubes cannot leave its tube (and thus the $L$-$\delta$-chain)
by $3\delta$ or more. On the other hand, none of the inner blocks' agents is able to even cross the preceeding or the following block of six tubes,
since their total length is larger than its budget.

\subparagraph{Arc Replacement of $\zeta$-Segments.} 
Next we decouple any pair of clause agents (stationed at nodes $c_{i,j},\, c_{j,k}$) that can directly go to the same literal node $l_{j,c}$ 
(so as not to allow them to transport the message via clause node with their augmented budgets, depicted in red in Figure~\ref{fig:replacements} (left)).
We replace the adjacent $\zeta$-segment by two small arcs which represent alternative ways over which the message can be transported.
Each arc consists of one $L$-$\delta$-chain and of one $\zeta$-segment, see Figure~\ref{fig:replacements}.\\
The \emph{inner arc} begins with the $\zeta$-segment -- whose beginning $l_{j,c}^i$ can be reached through an edge of length $1$ by the first clause agent (stationed at $c_{i,j}$) --
and ends with the $L$-$\delta$-chain.
The \emph{outer arc} first has the $L$-$\delta$-chain and then the $\zeta$-segment. The node in between these two parts, denoted by $l_{j,c}^k$, 
is connected via an edge of length $1$ to the second clause agent's starting position $c_{j,k}$.\\

We conclude the replacement with three remarks: Firstly, it is easy to see that the described operation respects the planarity of the graph. Secondly, we
are able to give values for $L$ and $\delta$ in the next subparagraph such that a single clause agent is still both necessary and (together with
agents inside the newly created adjacent $L$-$\delta$-chain) sufficient to transport a message over one of the parallel arcs from left to right. Finally, the
clause agent starting at $c_{i,j}$ is no longer able to meet the clause agent starting at $c_{j,k}$.

\begin{figure}[t!]
	\centering
	\includegraphics[width=\linewidth]{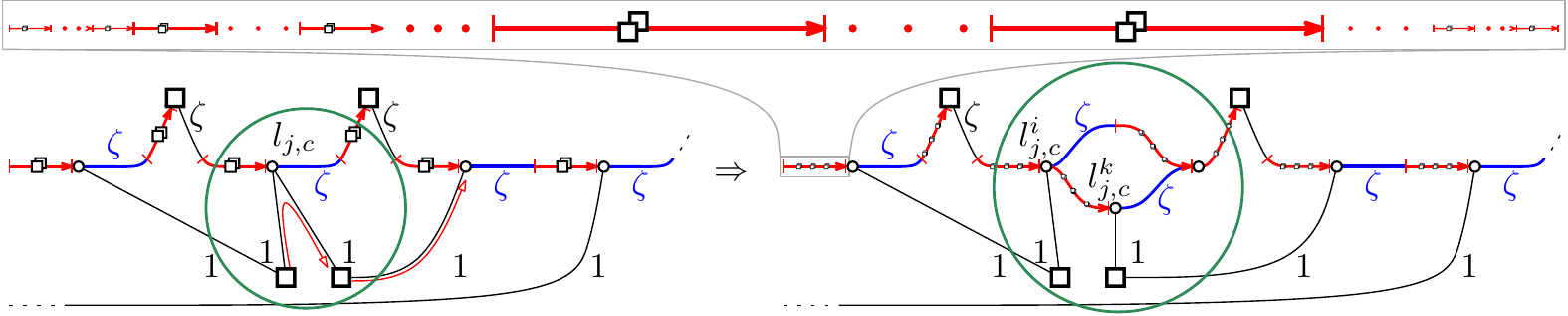}
	\caption{(top-to-bottom) We replace $\delta$-tubes in $G(F)$ by $L$-$\delta$-chains in $G'(F)$.\newline 
	(left-to-right) We replace each $\zeta$-segment connected to two clause agents by two parallel arcs.}
	\label{fig:replacements}
\end{figure}

\subparagraph{Budgets and Edge Weights.} Recall that our agents have the following budgets: \emph{separating agents} have a budget according to their
position in the $L$-$\delta$-chain, \emph{variable agents} a budget of $2\zeta$ and \emph{clause agents} a budget of $2(1+\zeta)$.\footnote{In the
\nonreturning version, variable agents have a budget of $\zeta$ and clause agents a budget of $1+\zeta$.} Now these budgets are
$\gamma$-resource-augmented, with $\gamma < 3$. We would like to prevent clause and variable agents from crossing $L$-$\delta$-chains or even meeting
inside of them, hence we set $L:= 9$, which shall exceed the augmented budget of every agent by a factor of more than 2. Furthermore we
don't want separating agents to help out too much outside of their residential chain, hence we set $\delta := \tfrac{\zeta}{9}$. A resource-augmented
separating agent can thus walk only as far as $3\delta = \tfrac{\zeta}{3}$ to the outside of the tube. In particular, separating agents cannot
transport the message over a $\zeta$-segment.

Next we choose $\zeta$ such that an augmented clause agent stationed at a clause node $c_{i,j}$ is not able to transport the message from $l_{i,c}$ to $l_{j,c}$, 
not even in collaboration with the separating agents that can reach the two literal nodes. 
We set $\zeta := \smash{\tfrac{\varepsilon}{6-\varepsilon}}$. 
The edges $\left\{ c_{i,j},\, l_{i,c} \right\},\, \left\{ c_{i,j},\, l_{j,c} \right\}$ have length $1$. 
In each edge, separating agents can help by at most $3\delta = \tfrac{\zeta}{3}$, leaving at least a distance of $1-\tfrac{\zeta}{3}$ for the clause agent to cover.
First note that for $0 < \varepsilon < 1$, we have $\zeta = \smash{\tfrac{\varepsilon}{6-\varepsilon} < \tfrac{\varepsilon}{5} < \tfrac{2\varepsilon}{3}}$ and $(6-\varepsilon) > 3(2-\varepsilon)$.
Hence a $\gamma$-resource-augmented clause agent has a budget of only
$\gamma\cdot 2(1+\zeta) = 2(2-\varepsilon)(1+\zeta) = 2( 2-\varepsilon+\tfrac{(2-\varepsilon)\varepsilon}{6-\varepsilon} ) < 2(2-\tfrac{2\varepsilon}{3}) < 2(2-\zeta) < 4\cdot (1-\tfrac{\zeta}{3})$, 
and thus cannot transport the message via its clause node and return home in the end.\footnote{For non-returning agents we use (for $\varepsilon < 2$) the inequalities: 
	$\zeta = \tfrac{\varepsilon}{6-\varepsilon} < \tfrac{\varepsilon}{4} < \tfrac{\varepsilon}{2}$ and $(6-\varepsilon) > 2(3-\varepsilon)$.
	Hence a non-returning $\gamma$-resource-augmented clause agent has a budget of 
	$\gamma(1+\zeta) = (3-\varepsilon)(1+\zeta) = 3-\varepsilon + \tfrac{(3-\varepsilon)\varepsilon}{6-\varepsilon} < 3-\tfrac{\varepsilon}{2} < 3-\zeta = 3\cdot (1-\tfrac{\zeta}{3})$, 
	and thus cannot transport the message via its clause node.\label{foot:3-eps-bounds}}	

\begin{restatable}[Resource-augmented Returning BudgetedDelivery]{lemma}{lemreshardness}
	A planar 3CNF $F$ is satisfiable if and only if it is possible to deliver a message with $(2-\varepsilon)$-resource-augmented agents from $s$ to $t$ in the corresponding delivery graph $G'(F)$, 
	such that the agents are still able to reach their starting point in the end.
	\label{lem:resource-augmented-budget-with-return}
\end{restatable}
\begin{proof}[Proof (Sketch).] 
	We follow the ideas of the proof of Lemma~\ref{lem:budget-with-return}, and use the modifications to the graph structure and the weights presented in this section. See Appendix~\ref{app:resource-augmentation-hardness} for a detailed proof.	
\end{proof}

\noindent The same arguments work for \nonreturning \ourProblem as well,  
if we replace the inequalities for returning $(2-\varepsilon)$-resource-augmented agents with the corresponding inequalities for non-returning $(3-\varepsilon)$-resource-augmented agents, 
given in Footnote~\ref{foot:3-eps-bounds}.

\noindent Compare the new delivery graph $G'(F)$ with the original graph $G(F)$. The only topological changes we introduced with our
replacements were the parallel arcs replacing the $\zeta$-segments reachable by two clause nodes. We have already seen that this change respected the
planarity of the delivery graph. Relevant changes to the edge weights and agent numbers, on the other hand, were added by replacing $\delta$-tubes
with $L$-$\delta$-chains: Each chain consists of blocks of six $\delta$-tubes of exponentially increasing size, hence we need a logarithmic number of
tubes per chain, namely $\bigO\left( \log \tfrac{L}{\delta}\right)$ many. We have fixed the values of $L$ and $\delta$ to $L=9$ and $\delta = \tfrac{\zeta}{9}$.
With $\zeta^{-1} = \tfrac{9}{\varepsilon} - 1 \in \Theta(\varepsilon^{-1})$ we get $\bigO\left( \log \tfrac{L}{\delta}\right) = \bigO(\log(\zeta^{-1})) =
\bigO(\log(\varepsilon^{-1}))$ many agents per chain. The number of chains is clearly polynomially bounded by the number of variables and clauses and the edge
weights depend on $\varepsilon$ only as well. Hence we conclude:

\begin{theorem}[Inexistence of a better resource augmentation for BudgetedDelivery]
	There is no polynomial-time $(2-\varepsilon)$-resource-augmented algorithm for \returning \ourProblem and no
	$(3-\varepsilon)$-resource-augmented algorithm for \nonreturning \ourProblem, unless $\mathrm{P}=\mathrm{NP}$. 
	\label{thm:resource-augmentation-hardness}
\end{theorem}

\section{Conclusions}
\label{sec:discussion}

We gave a polynomial time algorithm for the returning variant of the problem on trees, as well as a best-possible resource-augmented algorithm for general graphs.
On the other hand, we have shown that \ourProblem is \NP-hard, even on planar graphs and even if we allow resource augmentation.
Our bounds on the required resource augmentation are tight and complement the previously known algorithm~\cite{DDalgosensors13} for the non-returning case.

Our results show that \ourProblem becomes hard when transitioning from trees to planar graphs.
It is natural to investigate other causes for hardness.
Chalopin et al.~\cite{DDalgosensors13} gave a polynomial algorithm for the \nonreturning version under the assumptions that (i) the order in which the agents move is fixed and (ii) the message can only be handed over at vertices.
Using a dynamic program, we are able to drop assumption (ii), allowing handovers within edges.
Our result holds for both versions of \ourProblem.

\begin{theorem}
	\ourProblem is solvable in time $\bigO(k(n+m)(n\log n + m))$ if the agents are restricted to a fixed order in which they move.
	\label{thm:budget-with-schedule}
\end{theorem}

\begin{proof}
	If there is a feasible schedule, we can compute it in a breadth-first search-like fashion where we proceed agent by agent and update reachable regions of the
	graph on-the-fly:
	Each agent can either not help in the schedule or it can transport the message from a pickup location to a drop-off location. We show that we
	can restrict drop-offs to meaningful places such that for each agent the set of all possible pickup locations is bounded by $n + m$.
	This limitation to only one of the potentially infinitely many handover points inside each edge allows us to use dynamic programming and to proceed by induction: 
	
	Denote the agents in the schedule order by $a_1, \ldots, a_{\ell}$.
	The first agent $a_1$ can pick up the message at $s$ only, hence there is only one possible pick-up location. If $a_1$ wants to drop off the
	message at a vertex, there are at most $n$ choices of where to do so. We mark all the vertices which $a_1$ can reach from $s$ while still
	being able to return home. Now assume $a_1$ wants to drop off the message inside an edge $e = \{u,v\}$. This means that $e$ can be reached by
	$a_1$, hence without loss of generality the vertex $u$ is marked. If $v$ is marked as well, then $a_1$ should \emph{not} drop the message
	inside $e$, since the message has to be picked up later, which could just as well be done at either $u$ or $v$. Otherwise $a_1$ should bring
	the message \emph{as far as possible into the edge} (since if a later agent $a_i$ wants to pick up the message, it can pick it up at $u$ or come in
	via $v$). We mark this point inside the edge and store its distance from $u$. We now restrict ourselves to these at most $n+m$ described drop-off
	locations.

	An agent $a_i, \ i > 1$ can pick up the message at $s$ or at any previous drop-off location. By induction there are at most $n+m$ many
	such locations. Now we first check whether $a_i$ can pick up the message somewhere and deliver it to any not yet marked vertex. If so, we mark
	this vertex (and the number of marked vertices stays at most $n$). Next we check whether $a_i$ can bring the message into an edge $e = \{u,v \}$
	for which (without loss of generality) $u$ is marked and $v$ is not. Check whether the point inside the edge which is furthest from $u$ -- and
	still can be reached by $a_i$ -- has larger distance to $u$ than a previously marked point. If so, delete the old point (if any) and mark the
	new point and store its distance from $u$. The number of marked edges stays at most $m$. 
	
	If at some point we mark the vertex $t$, we are done.
	Since each agent $i$ has at most $n+m$ pick-up locations to consider, we can compute all new marks by computing the ellipsoid $\mathcal{E}(p_i,l,B_i)$ for every old mark $l$, which we can do by running Dijkstra's shortest path algorithm once from $p_i$ and once from each old mark. Hence we require time $\bigO((n+m)\cdot(n\log n + m))$ per agent.
\end{proof}

\begin{corollary}
For a constant number of agents $k$, \ourProblem is solvable in time  $\text{poly}(n,m)$ by brute forcing the order of the agents.
\end{corollary}

An interesting open problem is to understand collaborative delivery of multiple messages at once.
For example, the complexity of the problem on paths remains open.
In this setting, it may be resonable to constrain the number of agents, the number of messages, or the ability of transporting multiple messages at once, in order to allow for efficient algorithms.
Also, in general graphs, the problem may not become easy if the order in which agents move is fixed.

\newpage
\bibliographystyle{plain}
\bibliography{budgetbib}

\newpage
\appendix

\section*{Appendix}

\section{NP-Hardness}
\label{app:np-hardness}

\begin{figure}[h!]
	\centering
	\includegraphics[width=\linewidth]{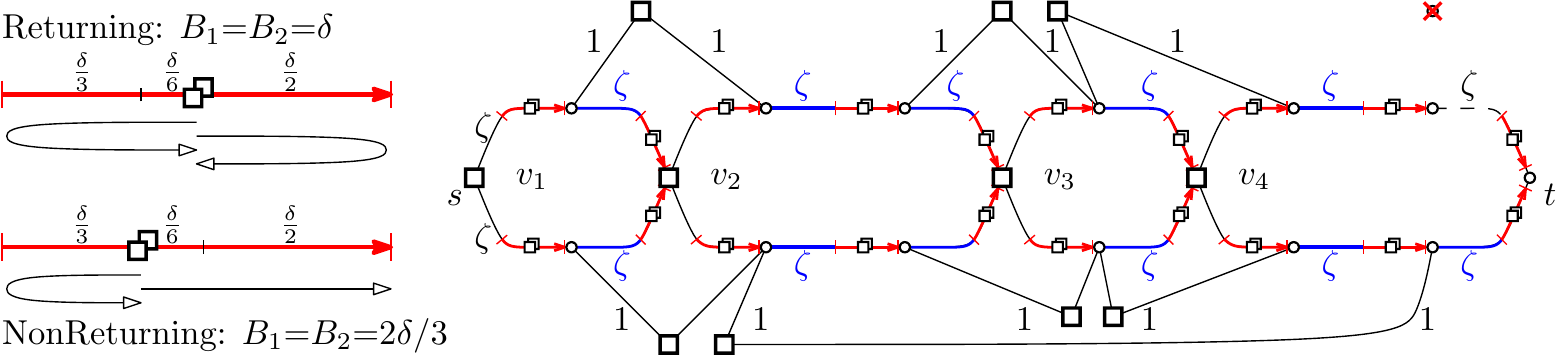}
	\caption{(left) Two examples of $\delta$-tubes for both versions of \ourProblem.\newline
		(right) Agent placement and edge weights on $G(F)$; agents are depicted by squares.}
	\label{fig:placement_weights2}
\end{figure}

\begin{proposition}
	\label{prop:SAT-schedule}
	If a planar 3CNF $F$ is satisfiable, then in the corresponding delivery graph $G(F)$,
	the agents can collectively deliver the message from $s$ to $t$ and return to their respective starting positions.
\end{proposition}
\begin{proof}
	Assume that there is a satisfiable assignment for $F$. Then the agents' actions are straightforward: 
	Each \emph{variable agent} placed on $v_i$ moves according to the variable assignment to $v_i$ by a $\zeta$-distance into either the \emph{true}-path $p_{i,\true}$ or the \emph{false}-path $p_{i,\false}$.
	For the message to be delivered to the next variable agent, it needs to be handed across $\delta$-tubes and $\zeta$-segments. 
	The former can always be done by the respective \emph{separating agents} residing inside the $\delta$-tube.
	It remains to be shown that the latter can be done by \emph{clause agents}.
	To this end, consider a clause $c$ which consists of $|c|$ literals. 

	If $|c|=1$ respectively $c = \ell(v_j)$ for some $j$, then there is no clause node in $G(F)$ at all (see the top right clause in Figure~\ref{fig:placement_weights2}). 
	No agent can reach the $\zeta$-segment adjacent to $l_{j,c}$, but this does not cause a problem, since by our assumption the literal $\ell(v_j)$ is satisfied and thus 
	the variable agent at $v_j$ chose to deliver the message via the opposite path $p_{j,\mathit{\ell(v_j)}}$.
	
	If $|c|=2$, then there is one clause agent on a single clause node $c_{i,j}$ which is connected to the internal path nodes $l_{i,c}$ and $l_{j,c}$ (see the top left clause in Figure~\ref{fig:placement_weights2}).
	Both have adjacent $\zeta$-segments which correspond to the literals $\overline{\ell(v_i)}, \overline{\ell(v_j)}$. 
	By assumption, at least one literal -- without loss of generality $\ell(v_j)$ -- is satisfied, and since the variable agent choosing the assignment for $v_j$
	thus takes the ``opposite'' path $p_{j,\mathit{\ell(v_j)}}$, the $\zeta$-segment corresponding to  $\overline{\ell(v_j)}$ does not need to be crossed while delivering the message. 
	If the other literal is not satisfied, then the clause agent is needed at the corresponding $\zeta$-segment, otherwise it can stay at its place of origin.
	
	If $|c|=3$, we have three literals/$\zeta$-segments and we have two clause nodes with one agent each (see the top center clause in Figure~\ref{fig:placement_weights2}). 
	One is connected to the first and the second $\zeta$-segment, the other to the second and third. 
	Collectively the two agents can reach every possible pair of segments out of the three $\zeta$-segments. 
	At least one literal $\ell(v_j)$ is satisfied. To each of the remaining $\zeta$-segments we can therefore send one agent.  
	Moving to the path needs 1 unit of energy (and so does returning to the clause node). Hence the agent can cover a remaining distance of
	$\zeta$, which is sufficient to transport the message to the next $\delta$-tube.
\end{proof}
\newpage

\begin{proposition}
	\label{prop:np-no-clause}
	It is not possible to deliver a message from $s$ to $t$ via any clause node of the delivery graph $G(F)$.
\end{proposition}
\begin{proof}
	For the sake of contradiction assume that the message is transported via a clause node $c_{i,j}$ which connects to the internal path nodes $l_{i,c}$ and $l_{j,c}$.
	Except for the clause agent stationed at $c_{i,j}$, no other agent can move further than $\zeta = \tfrac{1}{8}$ into each of the two edges $\left\{ l_{i,c},\, c_{i,j} \right\}$ and $\left\{ l_{j,c},\, c_{i,j} \right\}$.
	Hence the clause agent stationed at $c_{i,j}$ needs to cover in each edge a distance of $2(1-\zeta)$ (to go back and forth), 
	hence for both edges it needs an energy of at least $2\cdot 2(1-\zeta) = 4\cdot \tfrac{7}{8} = \tfrac{7}{2}$. 
	However, the agent has a budget of only $2(1+\zeta) = 2\cdot \tfrac{9}{8} < \tfrac{7}{2}$, yielding a contradiction to the message being transported over the clause 
	node.\footnote{Non-returning agents need to cover only one of the edges twice (to go back and forth), hence they need an energy of at least $3(1-\zeta) = \frac{21}{8}$ versus a budget of $1+\zeta = \frac{9}{8}$, yielding a contradiction as well.}
\end{proof}

\begin{proposition}
	\label{prop:np-gaps}
	To deliver the message over a $\zeta$-segment adjacent to a variable node $v_i$, we need the \emph{variable agent} with starting position $v_i$.
	To deliver the message over a $\zeta$-segment adjacent to a literal node $l_{j,c}$, we need a \emph{clause agent} with starting position $c_{i,j}$ or $c_{j,k}$.
\end{proposition}
\begin{proof}
	Recall that $\delta = \tfrac{4\zeta}{3}$.
	Separating agents inside $\delta$-tubes can neither single-handedly nor collectively (starting from both sides) transport the message over a $\zeta$-segment, 
	since they can move outside of their residential tube by at most $\tfrac{\delta}{3} < \tfrac{\zeta}{2}$.
	Furthermore variable agents and clause agents can move on a $\true$- or $\false$-path inside an interval of size at most $\zeta < \delta$, hence they can't cross a $\delta$-tube.
	Thus to transport the message over a $\zeta$-segment adjacent to a variable node $v_j$, we need the variable agent placed on $v_j$.
	On the other hand, transporting the message over a $\zeta$-segment adjacent to the internal path node $l_{j,c}$ needs a clause agent of clause $c$. 
	If $c$ has two clause nodes $c_{i,j}, c_{j,k}$, either of the two clause agents will do.
\end{proof}

\lemnphardness*
\begin{proof}
	``$\Rightarrow$'' This direction has been shown in Proposition~\ref{prop:SAT-schedule}.

	``$\Leftarrow$''  Assume that the message can be delivered from $s$ to $t$. 
	From Proposition~\ref{prop:np-no-clause} it follows that the message has to be transported through the $\true$- and $\false$-paths. 
	Without loss of generality, the message must move monotonously through the paths $p_{i,\true}$ or $p_{i,\false}$. 
	By Proposition~\ref{prop:np-gaps} we know that for each $\zeta$-segment that the message is delivered over, we need either the corresponding variable agent or the corresponding clause agent.
	It remains to show that we have enough clause agents for the task:

	Each clause with $|c|$ literals ``owns'' only $|c|-1$ clause agents and thus must have at least one satisfied literal 
	(otherwise the $|c|-1$ clause agents would not be sufficient to help in all corresponding $\zeta$-segments). 
	Hence we can read a satisfiable variable assignment for the \textsc{Planar3SAT} instance directly from the choice of the variable agents 
	(which each pick the adjacent $\true$- or the adjacent $\false$-path).
\end{proof}

\noindent It is easy to see that the same arguments work for \nonreturning \ourProblem as well, hence we immediately get the same statement for
the NonReturning version.
\begin{corollary}[Non-Returning BudgetedDelivery]
	A planar 3CNF $F$ is satisfiable if and only if it is possible to deliver a message in the corresponding delivery graph.	
	\label{cor:budget-without-return}
\end{corollary}
\newpage

\section{Hardness of Resource Augmentation}
\label{app:resource-augmentation-hardness}

\begin{figure}[h!]
	\centering
	\includegraphics[width=\linewidth]{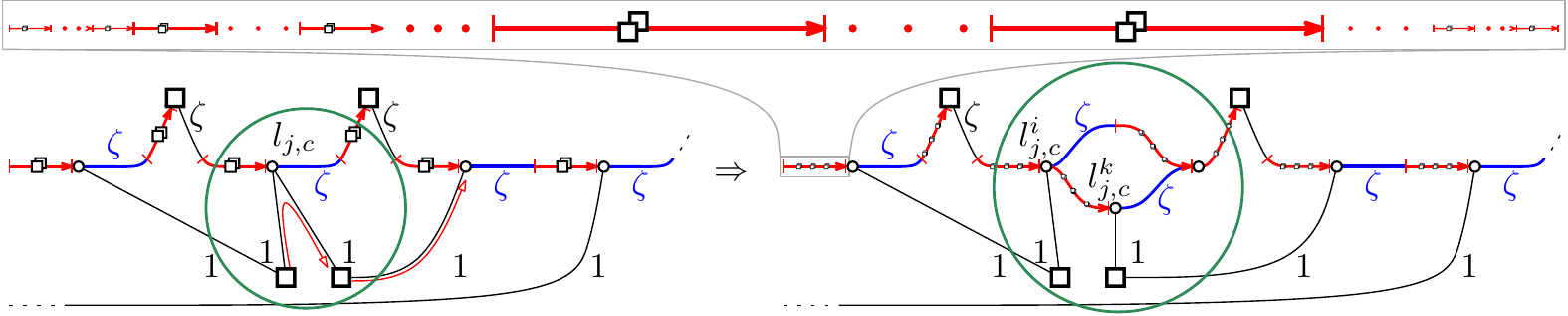}
	\caption{(top-to-bottom) We replace $\delta$-tubes in $G(F)$ by $L$-$\delta$-chains in $G'(F)$.\newline 
	(left-to-right) We replace each $\zeta$-segment connected to two clause agents by two parallel arcs.}
	\label{fig:replacements2}
\end{figure}

\begin{proposition}
	\label{prop:res-schedule}
	If a planar 3CNF $F$ is satisfiable, then in the corresponding delivery graph $G'(F)$,
	the agents can collectively deliver the message from $s$ to $t$ and return to their respective starting positions.
\end{proposition}
\begin{proof}
	In Proposition~\ref{prop:SAT-schedule} we have seen how the message can be transported in the original delivery graph $G(F)$. 
	In the modified delivery graph $G'(F)$, \emph{variable agents} and \emph{clause agents} do exactly the same as their counterparts in $G(F)$, and \emph{separating agents} help wherever needed.
\end{proof}

\begin{proposition}
	\label{prop:res-no-clause}
	It is not possible for $(2-\varepsilon)$-resource-augmented agents to deliver the message from $s$ to $t$ via any clause node of the delivery graph $G'(F)$.
\end{proposition}
\begin{proof}
	This has been shown in the main paper, we repeat the arguments here for the reader's convenience:
	We show that an augmented clause agent stationed at a clause node $c_{i,j}$ is not able to transport the message from $l_{i,c}$ to $l_{j,c}$, 
	not even in collaboration with the separating agents that can reach the two literal nodes. 
	The edges $\left\{ c_{i,j},\, l_{i,c} \right\},\, \left\{ c_{i,j},\, l_{j,c} \right\}$ have length $1$. 
	In each edge, separating agents can help by at most $3\delta = \tfrac{\zeta}{3}$, leaving at least a distance of $1-\tfrac{\zeta}{3}$ for the clause agent to cover.
	First note that for $0 < \varepsilon < 1$, we have $\zeta = \smash{\tfrac{\varepsilon}{6-\varepsilon} < \tfrac{\varepsilon}{5} < \tfrac{2\varepsilon}{3}}$ and $(6-\varepsilon) > 3(2-\varepsilon)$.
	Hence a $\gamma$-resource-augmented clause agent has a budget of only
	$\gamma\cdot 2(1+\zeta) = 2(2-\varepsilon)(1+\zeta) = 2( 2-\varepsilon+\tfrac{(2-\varepsilon)\varepsilon}{6-\varepsilon} ) < 2(2-\tfrac{2\varepsilon}{3}) < 2(2-\zeta) < 4\cdot (1-\tfrac{\zeta}{3})$, 
	and thus cannot transport the message via its clause node and return home in the end.\footnote{For non-returning agents we use (for $\varepsilon < 2$) the inequalities: 
		$\zeta = \tfrac{\varepsilon}{6-\varepsilon} < \tfrac{\varepsilon}{4} < \tfrac{\varepsilon}{2}$ and $(6-\varepsilon) > 2(3-\varepsilon)$.
		Hence a non-returning $\gamma$-resource-augmented clause agent has a budget of 
		$\gamma(1+\zeta) = (3-\varepsilon)(1+\zeta) = 3-\varepsilon + \tfrac{(3-\varepsilon)\varepsilon}{6-\varepsilon} < 3-\tfrac{\varepsilon}{2} < 3-\zeta = 3\cdot (1-\tfrac{\zeta}{3})$, 
		and thus cannot transport the message via its clause node.\label{foot:3-eps-bounds2}}
\end{proof}

\begin{proposition}
	\label{prop:res-original-budget}
	Assume that there is a schedule in which $\gamma$-resource-augmented agents ($\gamma < 3$) collectively deliver the message from $s$ to $t$ in the delivery graph $G'(F)$.
	Then in each $\zeta$-segment that the message is delivered over, the schedule uses the corresponding variable agent or the corresponding clause agent.
	Furthermore the schedule can be transformed into a feasible schedule with the original budgets.
\end{proposition}
\begin{proof}
	By Proposition~\ref{prop:res-no-clause} we know that the message cannot be transported over any of the clause nodes. 
	Recall that $\gamma$-resource-augmented separating agents inside $\delta$-tubes can neither single-handedly nor collectively (starting from both sides) 
	transport the message over a $\zeta$-segment, since they can move outside of their residential $L$-$\delta$-chain by at most $3\delta = \tfrac{\zeta}{3}$.
	Furthermore, clause agents and variable agents are not able to meet each other anywhere in the graph, since they are pairwise separated by at least one $L$-$\delta$-chain and do not have enough energy 
	(even with the resource-augmented budgets) to reach the middle point of one of these chains.
	
	In conclusion, we know that the message needs to be transported along the $\true$- and $\false$-paths and without loss of generality we assume that this happens in a strictly monotone movement. 
	Now in each $\zeta$-segment the message is transported across in the schedule with $(2-\varepsilon)$-resource-augmented budget, 
	a variable agent or a clause agent is necessary. Since these agents cannot meet each other, 
	such an agent must pick up the message from a separating agent of the preceeding $L$-$\delta$-chain and hand it over to a separating agent of the following $L$-$\delta$-chain. 
	Even with a non-augmented budget, said agent would be able to pick up and hand over the message at the end and the start of these chains. 
	Additionally, separating agents are able to transport the message from left to right across their $L$-$\delta$-chain without a resource augmentation of their budgets. 
	Together this yields a solution for delivery on $G'(F)$ without any resource-augmented budgets.
\end{proof}

\lemreshardness*
\begin{proof}
	``$\Rightarrow$'' This direction has been shown in Proposition~\ref{prop:res-schedule}.

	``$\Leftarrow$''  Assume that the message can be delivered from $s$ to $t$. 
	From Proposition~\ref{prop:res-no-clause} it follows that the message has to be transported through the $\true$- and $\false$-paths.
	By Proposition~\ref{prop:res-original-budget} we know that the schedule thus can be transformed into a schedule of agents with non-augmented budgets, 
	where in  each $\zeta$-segment that the message is delivered over, the corresponding variable agent or the corresponding clause agent is used.

	We show that there is a bijective mapping into a feasible schedule in the original graph $G(F)$, which by Lemma~\ref{lem:budget-with-return} 
	gives us a satisfiable assignment for $F$. Consider the movement of the individual agents:

	First of all, we let every \emph{variable agent} in $G(F)$ do the same work as its counterpart in $G'(F)$ and vice versa. 
	Now consider the \emph{separating agents} of any $\delta$-tube in $G(F)$ which corresponds to a $L$-$\delta$-chain in $G'(G)$. 
	We let these agents collectively transport the message from left to right over their $\delta$-tube in $G(F)$ 
	if and only if the agents in the corresponding chain in $G'(F)$ transport the message over their $L$-$\delta$-chain.
	Finally, we let corresponding \emph{clause agents} in both graphs go to the same $\zeta$-segment (both to their first segment, both to their second segment, or both to neither). 
	Hence agents in $G(F)$ can just ``copy'' the movements of their respective counterparts in $G'(F)$.
\end{proof}

\noindent It is easy to see that the same arguments work for \nonreturning \ourProblem as well: In the proof, we simply replace the use of
Lemma~\ref{lem:budget-with-return} by referring to Corollary~\ref{cor:budget-without-return} and replace the estimations in the proof of Proposition~\ref{prop:res-no-clause}
with the corresponding estimations for \nonreturning $(3-\varepsilon)$-resource-augmented agents, given in Footnote~\ref{foot:3-eps-bounds2}.

\begin{corollary}[Resource-augmented Non-Returning BudgetedDelivery]
	A planar 3CNF $F$ is satisfiable if and only if it is possible to deliver a message with $(3-\varepsilon)$-resource-augmented
	agents from $s$ to $t$ in the corresponding delivery graph $G'(F)$.	
	\label{cor:resource-augmented-budget-without-return}
\end{corollary}
\newpage

\end{document}